\documentclass[letterpaper, 10 pt, conference]{ieeeconf}  

\IEEEoverridecommandlockouts                              

\usepackage{graphics} 
\usepackage{epsfig} 
\usepackage{mathptmx} 
\usepackage{times} 
\usepackage{amsmath} 
\usepackage{amssymb}  
\usepackage[usenames]{color}
\usepackage{colortbl}
\usepackage[german]{babel}

\newtheorem{assumption}{Assumption}
\newtheorem{Proposition}{Proposition}

\def\rea{\mathbb{R}}

\title{\LARGE \bf
Parameter identification algorithm for a LTV system with partially unknown state matrix
}

\author{Olga Kozachek$^{1}$,  Nikolay Nikolaev$^{1}$, Olga Slita$^{1}$ and Alexey Bobtsov $^{1}$
\thanks{$^{1}$ Olga Kozachek, Alexey Bobtsov,  Nikolay Nikolaev and Olga Slita are with Department of Control Systems and Robotics, ITMO University, Kronverkskiy av. 49, Saint-Petersburg, 197101, Russia
        {\tt\small oakozachek@mail.ru, bobtsov@mail.ru, nikona@yandex.ru, o-slita@yandex.ru.}}%
}

\begin{document}

\maketitle
\thispagestyle{empty}
\pagestyle{empty}

\begin{abstract}

In this paper an adaptive state observer and parameter identification algorithm for a linear time-varying system is developed under condition that the state matrix of the system contains unknown time-varying parameters of a known form. The state vector is observed using only output and input measurements without identification of the unknown parameters. When the state vector estimate is obtained, the identification algorithm is applied to find unknown parameters of the system.

\end{abstract}

\section{INTRODUCTION}

One of the important problems in dynamical systems control is observation of the system state. In some cases it is possible to use sensors. However, sometimes the state vector cannot be measured directly. So state observation algorithms are applied, when it is impossible to place sensors to measure all of the state variables,. 

There are many algorithms that allow to obtain state vector estimate when the parameters of the system are known (for example, Kalman filter \cite{Kalman}, Luenberger observer \cite{Luenberger}, generalized parameter estimation based observer \cite{b18}). However, these methods can not be applied when the system contains unknown parameters. The state observer design problem becomes even more complicated for the case when the unknown parameters are time-varying.

Many researchers developed observers for linear time-varying systems (LTV) with unknown parameters. In \cite{Ghousein} a refrigeration system described by a second order hyperbolic partial differential equations (PDE) coupled with LTV ordinary differential equations (ODE) is considered. The problem of the PDE states, the ODE states and the unknown parameter estimation is solved. An adaptive nonlinear neural observer based approach for state of charge estimation of lithium-ion batteries is proposed in \cite{Shen}. Another example of a state observer is proposed in \cite{Rauh}, where an interval approach is used to solve a problem of parameter identification and observer design for distributed heating systems. Nevertheless, all of the methods mentioned above can be only applied for very specific systems. 

The problem of the state observer and parameter identification algorithms design in general for LTV systems with unknown time-varying parameters is still widely open. There are not that many solutions of this problem known at the moment (see \cite{Pyrkin}, \cite{Gerasimov}). In this work an adaptive state observer based on the output measurements for a LTV system with unknown time-varying parameters is developed. Combining the proposed state observer with algorithms that were described in previous papers of the authors (such as \cite{ICUMT}, \cite{IFAC}) it is possible to obtain state vector and unknown parameters estimates.

\section{PROBLEM FORMULATION}

We consider the following time-varying system
\begin{align}
	\label{sys}
	\dot{x}(t)&= A(t)x + B(t)u,\\ 
	y(t)&=Cx(t),
	\label{out}
\end{align}
where $x(t) \in \mathbb{R}^n$ is an unmeasured state vector, $y(t) \in \mathbb{R}$ is a measured output, $u(t) \in \mathbb{R}$ is a known input, $A(t) \in \mathbb{R}^{n\times n}$ is a partially known state matrix containing unknown time-varying parameters,  $B(t) \in \mathbb{R}^{n \times 1}$ and $C \in {\mathbb{R}}^{1 \times n}$ are known matrices and their entries are assumed to be continuous and bounded.

\begin{assumption}
Let us assume that matrix $A(t)$ can be represented as a sum of known and unknown parts:
\begin{equation}
A(t) = A_0(t) + D(t),
\end{equation}
where $A_0(t)$ 	is a known matrix and $D(t) \in \mathbb{R}^{n\times n}$ is a matrix of unknown time-varying parameters of the system.
\end{assumption}

\begin{assumption}
We suppose the matrix $D(t)=D(\theta(t))$, where $\theta(t)\in \mathbb{R}^n$, such that only one component in the row can be nonzero and the component $d_{ij}=0$ for $ i=\overline{1,n}, i < j$. 

We suppose that $\theta_{i}(t)$ is an unknown time-varying sinusoidal function defined by the equation
\begin{equation} 
	\ddot{\theta}_{i} (t)=-{\omega_{i}}^{2}\theta_{i}(t), 
	\label{thet}
\end{equation}
where $ \omega_{i} > 0 $ is a constant unknown parameter.
\end{assumption}

\begin{assumption}
We assume that the state vector of the system satisfies the following:
\begin{equation*}
x_i^2(t)>0 \mbox{ for } i = \overline{1, n}.
\end{equation*}
\end{assumption}

In the previous papers \cite{ICUMT}, \cite{IFAC} authors solved the problem of unknown parameters identification and state observation for the case of the systems with delayed output measurements. However, there was a strict restriction made that $C = I_{n \times n}$, which means that the state vector is measured. In this paper the identification and observation problems are solved with no restrictions for the matrix $C$. This new result develops the previous algorithms for a wider class of systems and allows to estimate the unknown time-varying parameters when only the output vector is measured. A state observer for the system \eqref{sys}, \eqref{out} is proposed and an example of the application of the identification algorithm based on the obtained state variables estimates is provided.

\section{MAIN RESULT}

Solution of the state observation problem and unknown parameters identification is based on the several steps. On the first step we construct the state observer for time-varying system. On the second step we use autors preweous results to reconstruct unknown time-varying parameters.

\subsection{STATE OBSERVER DESIGN}

Let us consider the system \eqref{sys}, \eqref{out} taking into account the assumption 1. The system can be rewritten as
\begin{align}
	\label{sys_new}
	\dot{x} &=A(t)x + D(\theta(t))x + B(t)u,\\
	y &= Cx.	
	\label{out_new}
\end{align}

\begin{Proposition}
If for the system \eqref{sys_new}, \eqref{out_new} there exist matrices $N \in \rea^{n \times 1}$, $G \in \rea^{n \times n}$, $M \in \rea^{n \times n}$, $M_c \in \rea^{n \times n}$ and $L \in \rea^{n \times 1}$ and the following conditions are satisfied
\begin{align}
\label{con1}
B&-N-GCB=0, \\
\label{con2}
D&-GCD=0, \\
\label{con3}
A_0&-GCA_0=M, \\
\label{con4}
M_c &= M-LC,
\end{align}
with $L$ that ensure an exponential convergence of the $\tilde{x}=x - \hat{x}$ to zero in the system $\dot{\tilde{x}}=M_c\tilde{x}$, then the estimate of the state vector $\hat{x}$ can be obtained from the following observer:
\begin{equation}
\label{obs}
\dot{\hat{x}} = M\hat{x}+Nu+G\dot{y}+LC\tilde{x},
\end{equation}
that provides convergence of the state estimation error $\tilde{x}$ to zero:
\begin{equation}
\lim_{t \to \infty} \tilde{x} = \lim_{t \to \infty} (x-\hat{x}) = 0.
\end{equation}

\end{Proposition}
\begin{proof}
Let us consider the state estimation error $\tilde{x}$. Its derivative is defined as:
\begin{align}
\nonumber \dot{\tilde{x}} &= Ax+Dx+Bu-M\hat{x}-Nu-G\dot{y}-LC\tilde{x} = \\
\nonumber &= A_0x+Dx+Bu-M\hat{x}-Nu- \\
\nonumber &-GC(A_0x+Dx+Bu)-LC\tilde{x} = \\
\nonumber &= (A_0-GCA_0)x - M\hat{x}+(B-N-GCB)u +\\
&+ (D-GCD)x - LC\tilde{x}. 
\end{align}

According to the conditions  \eqref{con1}, \eqref{con2}, \eqref{con3}, \eqref{con4} the equation for the estimation error derivative $\dot{\tilde{x}}$ can be rewritten in the following form:
\begin{equation}
\label{err_x}
\dot{\tilde{x}}=Mx-M\hat{x} - LC\tilde{x}=M\tilde{x}- LC\tilde{x} = M_c\tilde{x}.
\end{equation}

According to \eqref{con4} the estimation error $\tilde{x}$ in \eqref{err_x} converges to zero exponentially. 
\end{proof}

However, the observer \eqref{obs} contains unknown derivative $\dot{y}$ of the output variable. Let us define a new variable $z$ as:
\begin{equation}
\label{z}
z = \hat{x}-Gy.
\end{equation}

From \eqref{z} we have
\begin{equation}
\hat{x}=z+Gy.
\end{equation}

The derivative of the new variable can be obtained in the following form:
\begin{align}
\label{dot_z}
\nonumber \dot{z}&=M\hat{x}+Nu+LC\tilde{x} \\
&=M(z+Gy)+Nu+LC\tilde{x}.
\end{align}

Let us also define the estimate of the output:
\begin{equation}
\hat{y}=C\hat{x}.
\end{equation}

Therefore, the estimation error $\tilde{y}$ can be obtained:
\begin{equation}
\label{tilde_y}
\tilde{y}=y-\hat{y}=Cx-C\hat{x}=C\tilde{x}.
\end{equation}

Transforming \eqref{dot_z} and substituting \eqref{tilde_y} we can obtain the state vector estimate from the following system:
\begin{equation}
	\begin{cases} 
		\label{obs_fin}
		\dot{z}=Mz+MGy+Nu+Ly-LC\hat{x},\\ 
		\hat{x}=z+Gy.
	\end{cases}
\end{equation}

Once the state estimate vector $\hat{x}$ is obtained, known state-based identification algorithms can be used to find the unknown time-varying parameters. It is important to notice that the choice of the algorithm is provided and restricted by the system properties and the form of the system matrices.

\subsection{TIME-VARYING PARAMETERS IDENTIFICATION}

In this subsection we will use the idea from \cite{ICUMT} and \cite{IFAC} to solve the main problem of the unknown parameter identification.
Taking into account the assumption 2 and the fact that the state vector estimate is obtained, we can write the first equation of the system \eqref{sys_new} in the following form:
\begin{align}
	\label{dot_x1}
	\dot{\hat{x}}_1= a_{11}\hat{x}_{1}+a_{12}\hat{x}_{2}+a_{13}\hat{x}_{3}+ \ldots + a_{1n}\hat{x}_{n} + \theta_1\hat{x}_1+ b_{1}u,
\end{align}
where $a_{ij}$ are elements of known matrix $A_0(t)$, $b_1$ is an element of known matrix $B(t)$ and signals $\hat{x}_i$ and $u$ are known. For the further simplification let us define new variable:
\begin{equation*}
h_i=a_{i1}\hat{x}_{1}+a_{i2}\hat{x}_{2}+a_{i3}\hat{x}_{3}+ \ldots + a_{in}\hat{x}_{n} + b_{i}u.
\end{equation*}
The \eqref{dot_x1} can be rewritten:
\begin{equation}
\label{x_h}
\dot{\hat{x}}_1=h_1 + \theta_1\hat{x}_1.
\end{equation}
For the estimation of the unknown parameter $\theta_1(t)$ we need to make several steps (see \cite{ICUMT}, \cite{IFAC}). At first, let us consider new variable
\begin{align}
	\label{V}
	V_1= \hat{x}_1^2.
\end{align}
Then for derivative of $V$ we have
\begin{align}
	\dot{V}_1 &= 2\hat{x}_1h_1 + 2\hat{x}_1^2\theta_1(t).
\end{align}

Let us define another variable as follows:
\begin{equation}
	\xi_1=\ln V_1.
\end{equation}

The derivative $\dot{\xi}_1$ can be written in the following form:
\begin{equation}
	\label{ksi}
	\dot{\xi}_1=\frac{\dot{V}_1}{V_1}=2\alpha_1+2\theta_1(t),
\end{equation}
where $\alpha_1=\frac{h_1}{\hat{x}_1}$.

Now the unknown parameter $\theta_1(t)$ can be obtained from \eqref{ksi}:
\begin{equation}
	\label{theta}
	\theta_1(t)=\frac{1}{2}\dot{\xi}_1-\alpha_1.
\end{equation}

Let us rewrite \eqref{thet} for $\theta_1$ in the operator form and substitute it into \eqref{theta}. Thus, we obtain:
\begin{equation}
	p^2[\frac{1}{2}\dot{\xi}_1-\alpha_1]=-\omega_1^2(\frac{1}{2}\dot{\xi}_1-\alpha_1).
\end{equation}

By applying an LTI filter $\frac{\lambda}{(p+\lambda)^3}$, where $\lambda>0$ we can transform the previous equation into linear regression model of the following form:
\begin{equation}
	Y_1 = \Phi_1 k_1,
\end{equation}
where:
\begin{align*}
	Y_1&=\frac{1}{2}\frac{\lambda p^3}{(p+\lambda)^3}[\xi_1]-\frac{\lambda p^2}{(p+\lambda)^3}[\alpha_1]; \\
	\Phi_1 &= \frac{1}{2}\frac{\lambda p}{(p+\lambda)^3}[\xi_1]-\frac{\lambda }{(p+\lambda)^3}[\alpha_1]; \\
	k_1 &= -\omega_1^2.
\end{align*}

The unknown parameter $k_1$ can be estimated using any proper algorithm for linear regression parameter estimation. For example, we can use gradient algorithm \cite{Ljung, Sastry}:
\begin{equation}
	\label{grad}
	\dot{\hat{k}}_1=- \gamma_1 \Phi_1 (\Phi_1 \hat{k}_1 - Y_1),
\end{equation}
where $\gamma_1>0$.

After that the unknown constant parameter $\omega_1$ can be calculated:
\begin{equation}
	\label{hat_omega}
	\hat{\omega}_1=\sqrt{\left| \hat{k}_1 \right|}.
\end{equation}

In the next step we consider the equation \eqref{thet}, which solution takes the form
\begin{equation}
	\label{theta_sol}
	\theta(t)_1=l_1^T\chi_1,
\end{equation}
where $\chi_1 = \begin{bmatrix}
	\sin(\omega_1 t) \\
	\cos(\omega_1 t)
\end{bmatrix}$, 
$l_1=\begin{bmatrix}
	l_{11} \\
	l_{12}
\end{bmatrix}$ - is an unknown constant vector defined by initial conditions of the equation \eqref{thet}.

The initially unknown constant parameter $\omega_1$ can be replaced by its estimate $\hat{\omega}_1$ obtained from \eqref{hat_omega}. After that the expression \eqref{theta_sol} can be substituted to \eqref{dot_x1}:
\begin{equation}
	\label{hat_x1}
	\dot{\hat{x}}_1=l_1^T \hat{\chi_1}\hat{x}_1 + h_1,
\end{equation}
where $\hat{\chi}_1=\begin{bmatrix}
	\sin ({\hat{\omega}_1}t) \\
	\cos ({\hat{\omega}_1}t)
\end{bmatrix}$.

This equation can be also transformed by applying LTI filter $\frac{\lambda_1}{p+\lambda_1},$ where $\lambda_1>0$. The obtained linear regression model can be written in the following form:
\begin{equation}
	\label{regression}
	q_1=\phi_1^T l_1,
\end{equation}
where:
\begin{align}
	\label{mathcalY}
	q_1 &=\frac{\lambda_1 p}{p+\lambda_1}[\hat{x}_1]-\frac{\lambda_1}{p+\lambda_1}[h_1]; \\
	\phi_1 &= \frac{\lambda_1}{p+\lambda_1} \begin{bmatrix}
		\hat{x}_1 \sin (\hat{\omega}_1t) \\
		\hat{x}_1 \cos (\hat{\omega}_1t)
	\end{bmatrix}.
\end{align}

For estimation of unknown parameters $l_{11}$ and $l_{12}$ we suggest using dynamical regression extention and mixing (DREM) technology \cite{Aranovskiy, Ortega} in the form it was considered in \cite{Bobtsov}. Then the observer for unknown parameters $l_{11}$ and $l_{12}$ can be written in the following form 
\begin{align}
	\dot{\mathcal{Y}}_1 &= - \lambda_2 \mathcal{Y}_1 +  \lambda_2\phi_1^\top q_1, \\
	\dot\Omega_1 &=- \lambda_2 \Omega_1 +  \lambda_2\phi_1^\top \phi_1,	\\
	\dot {\hat{l}}_1 &=-\gamma_2 \Delta_1 (\Delta_1 \hat{l}_1-\mathcal{Z}_1),
\end{align}
with $\lambda _2>0$ and $\gamma_2>0$, 
$\hat{l}_1=
\begin{bmatrix}
	\hat{l}_{11} \\ \hat{l}_{12}
\end{bmatrix}$ and with the definitions
\begin{align}
	\mathcal{Z}_1 &= \mbox{adj}\{{\Omega_1}\}\mathcal{Y}_1,\\
	\label{delta}
	\Delta_1&=\det\{{\Omega_1}\}.
\end{align}

If we use \eqref{theta_sol}, we can obtain the estimate of the unknown time-varying parameter $ \hat\theta_1(t) $ by substitution of $\hat{l}_1$ and $\hat{\omega}_1$.

After the first unknown parameter is obtained, every $\theta_i(t)$ can be found by the following algorithm.
First of all, let us write down the equation for $\dot{\hat{x}}_i$. According to the assumption 2, it takes the following form:
\begin{equation}
\label{hat_x_i}
\dot{\hat{x}}_i=h_i + \theta_i(t)\hat{x}_s,
\end{equation}
where $1 \leq s \leq i$.

In case $s=i$, the algorithm of the unknown parameter $\theta_i(t)$ estimation is corresponds the algorithm for $\theta_1(t)$ described above.

In case $s < i$, the expression of the unknown parameter can be obtained from \eqref{hat_x_i}:
\begin{equation}
\label{theta_i}
\theta_i(t) = \frac{\dot{\hat{x}}_i}{\hat{x}_s}-\frac{h_i}{\hat{x}_s}.
\end{equation}

Similarly to the algorithm for the $\theta_1(t)$, we substitute \eqref{theta_i} to \eqref{thet} and apply LTI filter $\frac{\lambda_i^3}{(p+\lambda_i)^3}$, where $\lambda_i>0$:
\begin{align}
\label{theta_i_filter}
\frac{\lambda_i^3p^2}{(p+\lambda_i)^3}[\frac{\dot{\hat{x}}_i}{\hat{x}_s}]&-\frac{\lambda_i^3p^2}{(p+\lambda_i)^3}[\frac{h_i}{\hat{x}_s}]= \\
\nonumber
&-\omega_i^2(\frac{\lambda_i^3}{(p+\lambda_i)^3}[\frac{\dot{\hat{x}}_i}{\hat{x}_s}]-\frac{\lambda_i^3}{(p+\lambda_i)^3}[\frac{h_i}{\hat{x}_s}]).
\end{align}

According to Swapping Lemma \cite{Jonsson}:
\begin{equation}
\label{lemma}
\frac{\lambda_i}{(p+\lambda_i)}[\frac{\dot{\hat{x}}_i}{\hat{x}_s}]=\frac{1}{x_{s}}\frac{\lambda_i p}{p+\lambda_i}x_{i} + \frac{\lambda_i}{p+\lambda_i}[\frac{\dot{x}_{s}}{x_{s}^2}\frac{p}{p+\lambda_i}[x_{i}]],
\end{equation}
where $\dot{\hat{x}}_s$ is assumed to be known as it can be obtained on previous steps.

After transforming \eqref{theta_i_filter} using \eqref{lemma}, linear regression model can be obtained:
\begin{equation*}
Y_i = \Phi_i k_i,
\end{equation*}
where:
\begin{align*}
Y_i = \frac{\lambda_i^2p^2}{(p+\lambda_i)^2}[\frac{1}{x_{s}}\frac{\lambda_i p}{p+\lambda_i}x_{i} &+ \frac{\lambda_i}{p+\lambda_i}[\frac{\dot{x}_{s}}{x_{s}^2}\frac{p}{p+\lambda_i}[x_{i}]]] - \\
&\frac{\lambda_i^3p^2}{(p+\lambda_i)^3}[\frac{h_i}{\hat{x}_s}] \\
\Phi_i = \frac{\lambda_i^2}{(p+\lambda_i)^2}[\frac{1}{x_{s}}\frac{\lambda_i p}{p+\lambda_i}x_{i} &+ \frac{\lambda_i}{p+\lambda_i}[\frac{\dot{x}_{s}}{x_{s}^2}\frac{p}{p+\lambda_i}[x_{i}]]] - \\
&\frac{\lambda_i^3}{(p+\lambda_i)^3}[\frac{h_i}{\hat{x}_s}] \\
k_i = &-\omega_{i}^2.
\end{align*}

The unknown parameter estimate $\hat{k}_i$ can be obtained using gradient algorithm. Then, using the obtained value we can calculate the unknown parameter estimate $\hat{\omega}_i$. The further process of the unknown parameter $\theta_i(t)$ estimation is completely similar to the algorithm of the first unknown parameter $\theta_1(t)$ estimation.
\section{NUMERICAL EXAMPLE AND SIMULATION}

Let us apply the proposed algorithm to the system \eqref{sys_new}, \eqref{out_new} with the following parameters
\begin{align*}
&A(t) = \begin{bmatrix}
0 &0.1-0.1\sin(t)+0.1\sin(5t)+1.5\cos(5t) \\
-0.1 &-1+0.5\cos(2t)
\end{bmatrix}; \\
&B(t)=\begin{bmatrix}
-1 \\
4
\end{bmatrix}; 
C(t) = \begin{bmatrix}
1 &1
\end{bmatrix}; 
D(t) = \begin{bmatrix}
\theta(t) &0 \\
0 &0
\end{bmatrix}.
\end{align*}
The unknown parameter $\theta(t)$ is defined by a differential equation \eqref{thet} with $l=\begin{bmatrix}
l_{1} \\
l_{2}
\end{bmatrix}=\begin{bmatrix}
3 \\
0.5
\end{bmatrix}, \chi = \begin{bmatrix}
\sin(3 t) \\
\cos(3 t)
\end{bmatrix}$.

The problem is to estimate the state vector of the system and the parameter $\theta(t)$ assuming that the constant parameters $l$ and $\omega$ are unknown. 

\subsection{State observer}

Let us consider matrices:
\begin{align*}
G=\begin{bmatrix}
1 \\
0
\end{bmatrix}, N = \begin{bmatrix}
-4 \\
4
\end{bmatrix}, L = \begin{bmatrix}
0.1 \\
0.5
\end{bmatrix}, \\
M = \begin{bmatrix}
0.1 &1-0.5\cos(2t) \\
-0.1 &-1+0.5\cos(2t).
\end{bmatrix}
\end{align*}

Since these matrices satisfy the conditions \eqref{con1}-\eqref{con4}, the state vector estimate $\hat{x}$ can be found by substituting these matrices into \eqref{obs_fin}.

\subsection{Time-varying parameter estimation}

Since the matrix $D(t)$ contains unknown parameter in the first row only, let us consider just the first equation of the system with the state vector estimate:
\begin{equation}
\label{hat_x}
\dot{\hat{x}}_1=\theta(t)\hat{x}_1+a_{12}\hat{x}_2-u,
\end{equation}
where $a_{12}=0.1-0.1\sin t+0.1\sin (5t)+1.5\cos (5t)$.

By applying an LTI filter $\frac{\lambda}{(p+\lambda)^3}$, $\lambda>0$, we can transform the previous equation into linear regression model of the following form
\begin{equation}
Y = \Phi k,
\end{equation}
where:
\begin{align*}
Y&=\frac{\lambda p^3}{(p+\lambda)^3}[\frac{1}{2}\xi]-\frac{\lambda p^2}{(p+\lambda)^3}[\alpha]; \\
\Phi &= \frac{\lambda p}{(p+\lambda)^3}[\frac{1}{2}\xi]-\frac{\lambda }{(p+\lambda)^3}[\alpha]; \\
k &= -\omega^2.
\end{align*}

For the estimation of the unknown parameter $k$ the gradient algorithm \eqref{grad} is used and for calculation of unknown parameter $\omega$ the equation \eqref{hat_omega} is used.

For estimation of unknown parameters $l_{1}$ and $l_{2}$ expression \eqref{hat_x1} was transformed into linear regression form \eqref{regression} and the unknown parameters were estimated using observer \eqref{mathcalY}...\eqref{delta}.

\subsection{Simulation results}
The described algorithm for the considered system was simulated in MATLAB Simulink. The real state vector $x$ and state vector estimate $\hat{x}$ are shown on the fig.~\ref{xx_hat}. The fig.~\ref{x_err} demonstrates the convergence to zero of the state estimation error $\tilde{x}$. 
   \begin{figure}[thpb]
      \centering
      \includegraphics[scale=0.55]{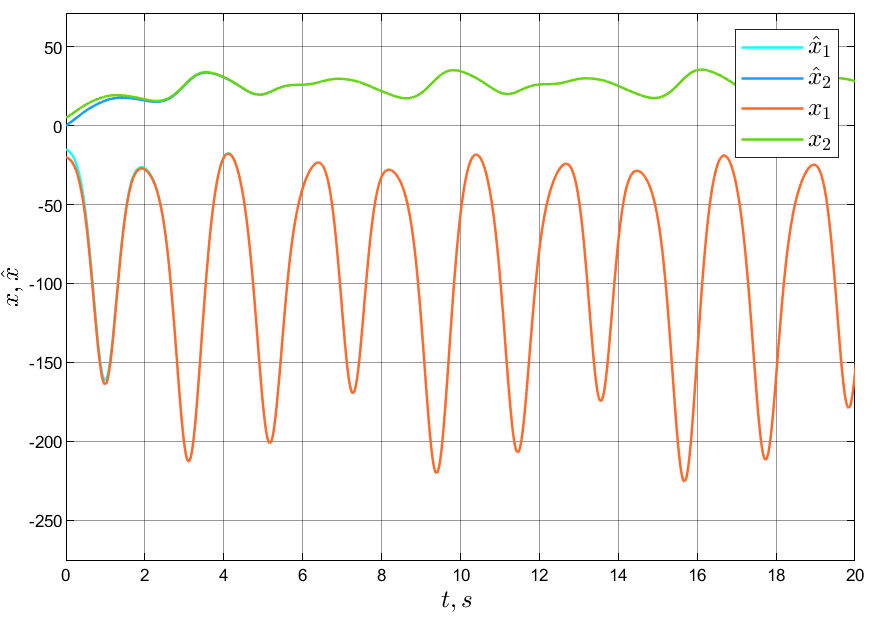}
      \caption{The transients of the state vector $x$ and its estimate $\hat{x}$}
      \label{xx_hat}
   \end{figure}
     
   \begin{figure}[thpb]
      \centering
      \includegraphics[scale=0.55]{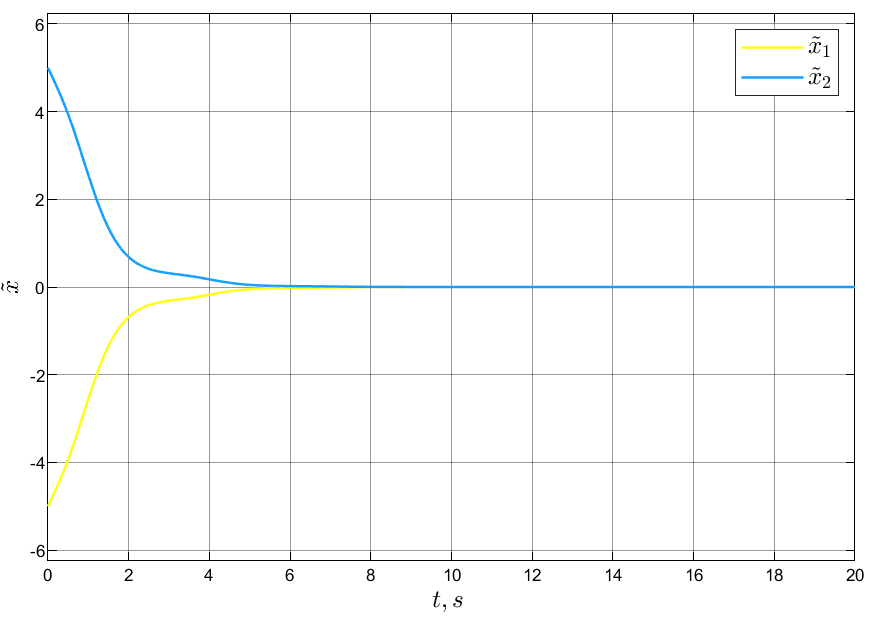}
      \caption{The transient of the state vector estimation error $\tilde{x}$}
      \label{x_err}
   \end{figure}
   
Identification of the unknown parameter $\theta(t)$ is based on the obtained state vector estimate. The first step of the identification process is estimation of the unknown constant parameter $\omega$. The estimation error $\tilde{\omega}=\omega-\hat{\omega}$ transient is demonstrated on the fig.~\ref{omega}. The error is converges to zero.

\begin{figure}[thpb]
      \centering
      \includegraphics[scale=0.55]{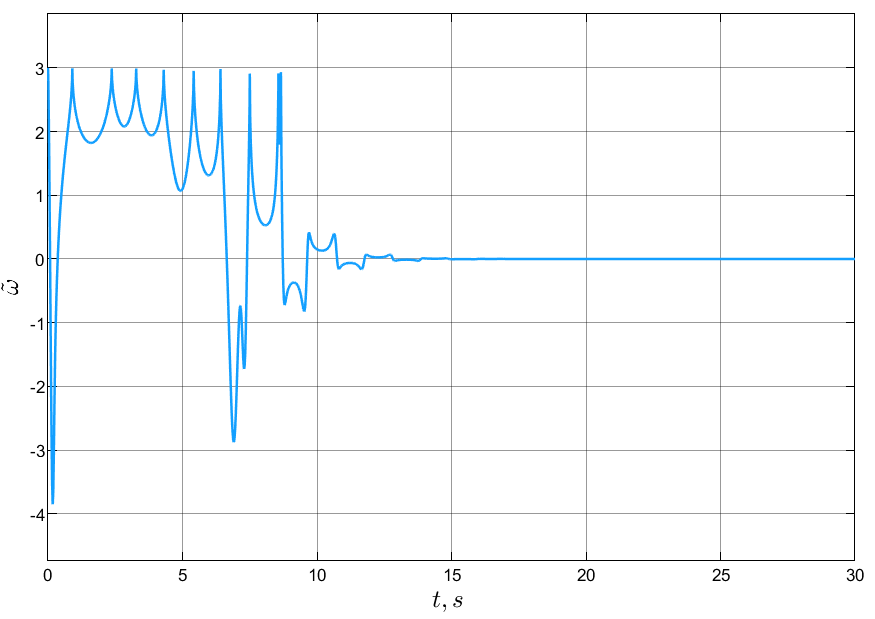}
      \caption{The transient of the estimation error $\tilde{\omega}$}
      \label{omega}
\end{figure}

After the estimate of the constant parameter $\hat{\omega}$ is obtained, the pair of the unknown parameters $l_1$ and $l_2$ is estimated. The estimation error $\tilde{l}l-\hat{l}$ transient is demonstrated on the fig.~\ref{a}. The error is converges to zero.

\begin{figure}[thpb]
      \centering
      \includegraphics[scale=0.5]{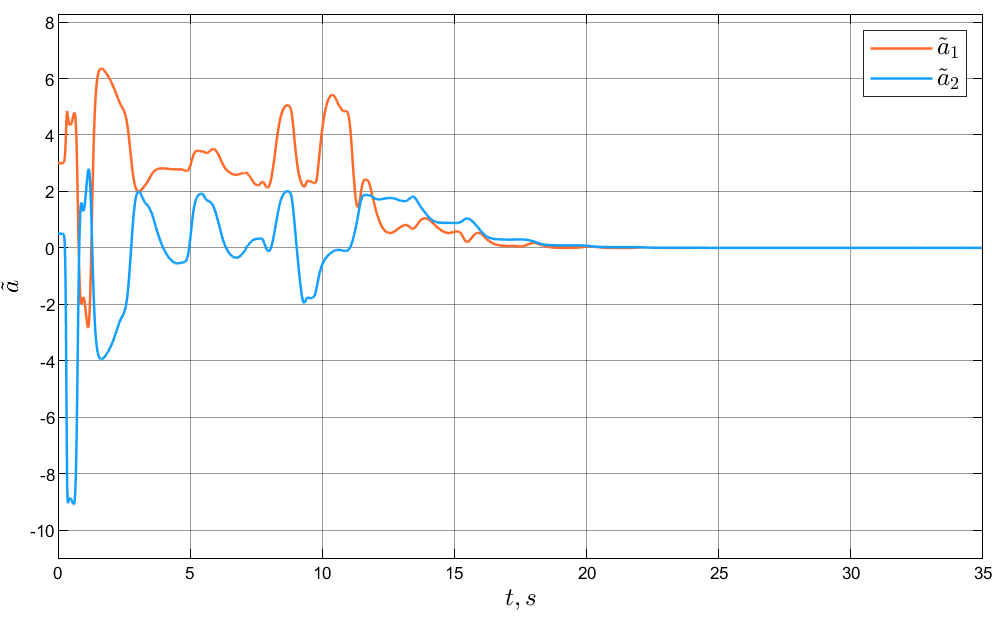}
      \caption{The transient of the estimation error $\tilde{a}$}
      \label{a}
\end{figure}

The parameters estimates $\hat{\omega}$ and $\hat{a}$ are used for estimation of the unknown parameter $\theta(t)$. The transients of the parameter $\theta(t)$ real value and its estimate $\hat{\theta}(t)$ are demonstrated on the fig.~\ref{theta_hat}. The fig.~\ref{theta_err} demonstrates the convergence to zero of the estimation error $\tilde{\theta}=\theta - \hat{\theta}$.  

\begin{figure}[thpb]
      \centering
      \includegraphics[scale=0.5]{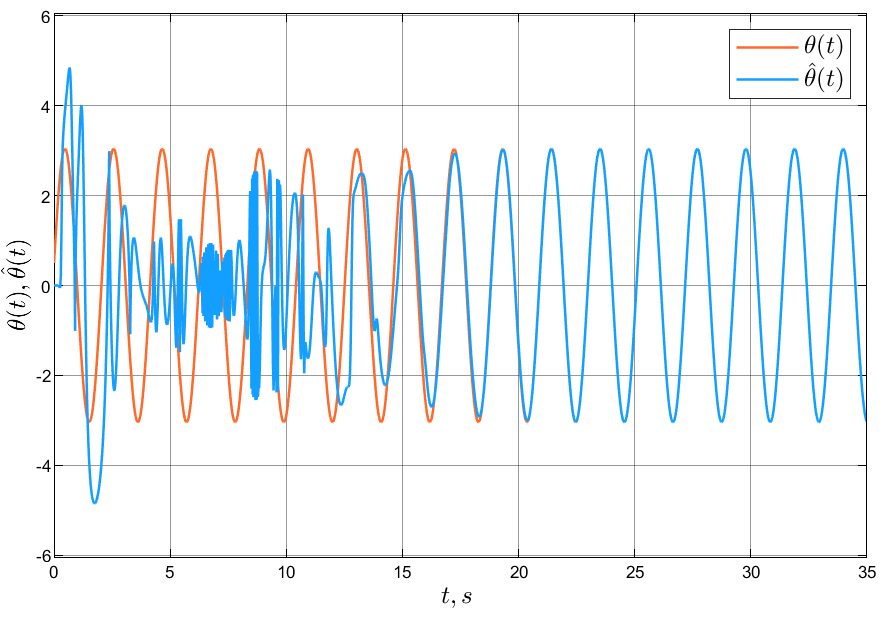}
      \caption{The transient of the parameter $\theta(t)$ and its estimate $\hat{\theta}(t)$}
      \label{theta_hat}
\end{figure}

\begin{figure}[thpb]
      \centering
      \includegraphics[scale=0.49]{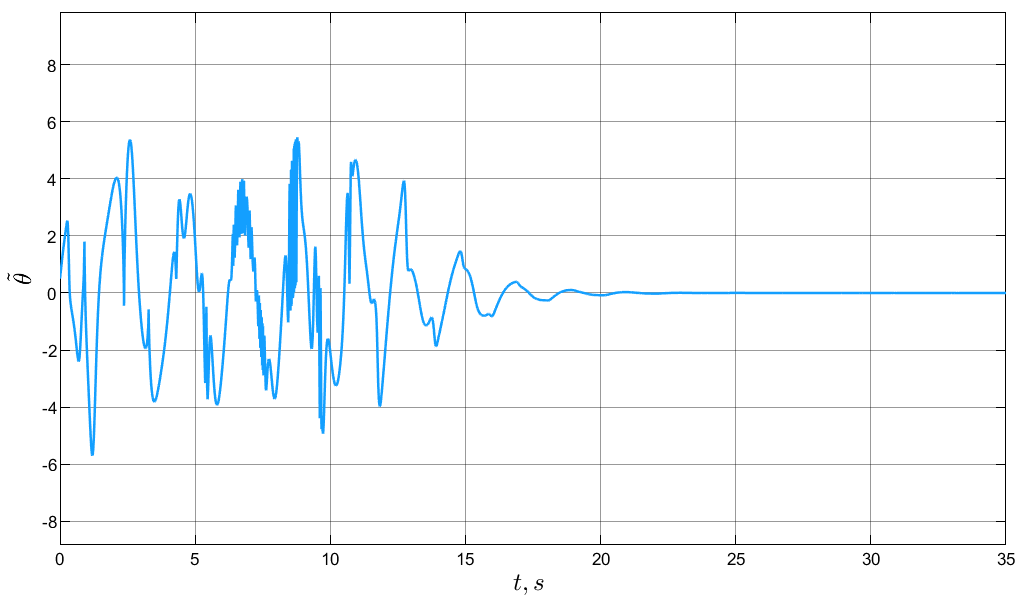}
      \caption{The transient of the estimation error $\tilde{\theta}$}
      \label{theta_err}
\end{figure}

\section{CONCLUSIONS}

An adaptive observer for a linear time-variant system \eqref{sys}, \eqref{out} with unknown time-varying parameters in the state matrix $A(t)$ was presented. The developed algorithm provides state vector estimates based on the known input signal and measured output signal. The estimation error converges to zero, that was demonstrated by simulation. The algorithm makes also possible to use state-based identification methods when the state vector is unknown. To demonstrate it, the identification algorithm that was proposed in \cite{ICUMT} was applied in this paper to estimate the unknown time-varying parameter of the system.

\addtolength{\textheight}{-12cm}   


\bibliography{ifacconf}

\begin{thebibliography}{10}
\providecommand{\url}[1]{#1}
\csname url@rmstyle\endcsname
\providecommand{\newblock}{\relax}
\providecommand{\bibinfo}[2]{#2}
\providecommand\BIBentrySTDinterwordspacing{\spaceskip=0pt\relax}
\providecommand\BIBentryALTinterwordstretchfactor{4}
\providecommand\BIBentryALTinterwordspacing{\spaceskip=\fontdimen2\font plus
\BIBentryALTinterwordstretchfactor\fontdimen3\font minus
  \fontdimen4\font\relax}
\providecommand\BIBforeignlanguage[2]{{%
\expandafter\ifx\csname l@#1\endcsname\relax
\typeout{** WARNING: IEEEtran.bst: No hyphenation pattern has been}%
\typeout{** loaded for the language `#1'. Using the pattern for}%
\typeout{** the default language instead.}%
\else
\language=\csname l@#1\endcsname
\fi
#2}}

\bibitem{ICUMT}
B.~obtsov, N.~Nikolaev, O.~Slita, O.~Kozachek, and O.~Oskina, ``Adaptive
  observer for a ltv system with partially unknown state matrix and delayed
  measurements,'' in \emph{2022 14th International Congress on Ultra Modern
  Telecommunications and Control Systems and Workshops (ICUMT)}, 2022, pp.
  165--170.

\bibitem{Kalman}
R.~E. Kalman, ``A new approach to linear filtering and prediction problems,''
  \emph{Journal of Basic Engineering}, vol.~81, no.~1, pp. 35--45, 1960.

\bibitem{Luenberger}
D.~G. Luenberger, ``Observing the state of a linear system,'' \emph{IEEE
  Transactions on Military Electronics}, vol.~8, no.~2, pp. 74--80, 1964.

\bibitem{b18}
R.~Ortega, A.~Bobtsov, N.~Nikolaev, and D.~D. J.~Schiffer, ``Generalized
  parameter estimation-based observers: Application to power systems and
  chemical–biological reactors,'' \emph{Automatica}, vol. 129, 2021.

\bibitem{Ghousein}
M.~Ghousein and E.~Witrant, ``Adaptive observer design for uncertain hyperbolic
  pdes coupled with uncertain ltv odes; application to refrigeration systems,''
  \emph{Automatica}, vol. 154, p. 111096, 2023.

\bibitem{Shen}
Y.~Shen, ``A robust method for state of charge estimation of lithium-ion
  batteries using adaptive nonlinear neural observer,'' \emph{Journal of Energy
  Storage}, vol.~72, p. 108480, 2023.

\bibitem{Rauh}
A.~Rauh, J.~Kersten, and H.~Aschemann, ``An interval approach for parameter
  identification and observer design of spatially distributed heating
  systems,'' \emph{IFAC PapersOnLine}, vol.~51, no.~2, pp. 337--342, 2018.

\bibitem{Pyrkin}
A.~Pyrkin, A.~Bobtsov, R.~Ortega, and A.~Isidori, ``An adaptive observer for
  uncertain linear time-varying systems with unknown additive perturbations,''
  \emph{Automatica}, vol. 147, p. 110677, 2023.

\bibitem{Gerasimov}
D.~Gerasimov, A.~Popov, N.~Hien, and V.~Nikiforov, ``Adaptive control of ltv
  systems with uncertain periodic coefficients,'' \emph{IFAC-PapersOnLine},
  vol.~56, no.~2, pp. 9185--9190, 2023, 22nd IFAC World Congress.

\bibitem{IFAC}
O.~Kozachek, A.~Bobtsov, and N.~Nikolaev, ``Adaptive observer for a nonlinear
  system with partially unknown state matrix and delayed measurements,''
  \emph{IFAC-PapersOnLine}, vol.~56, no.~2, pp. 8702--8707, 2023, 22nd IFAC
  World Congress.

\bibitem{Ljung}
L.~Ljung, \emph{System identification: theory for the users}.\hskip 1em plus
  0.5em minus 0.4em\relax New Jersey: Prentice Hall, 1999.

\bibitem{Sastry}
S.~Sastry, M.~Bodson, and J.~F. Bartram, \emph{Adaptive control: stability,
  convergence, and robustness}, 1990.

\bibitem{Aranovskiy}
S.~Aranovskiy, A.~Bobtsov, R.~Ortega, and A.~Pyrkin, ``Performance enhancement
  of parameter estimators via dynamic regressor extension and mixing,''
  \emph{IEEE Transactions on Automatic Control}, vol.~62, no.~7, pp.
  3546--3550, 2017.

\bibitem{Ortega}
R.~Ortega, S.~Aranovskiy, A.~A. Pyrkin, A.~Astolfi, and A.~A. Bobtsov, ``New
  results on parameter estimation via dynamic regressor extension and mixing:
  Continuous and discrete-time cases,'' \emph{IEEE Transactions on Automatic
  Control}, vol.~66, no.~5, pp. 2265--2272, 2021.

\bibitem{Bobtsov}
A.~Bobtsov, N.~Nikolaev, R.~Ortega, and D.~Efimov, ``State observation of
  affine-in-the-states time-varying systems with unknown parameters and delayed
  measurements,'' \emph{Third IFAC Conference on Modelling, Identification and
  Control of Nonlinear Systems}, pp. 124--129, 2021.

\bibitem{Jonsson}
U.~J"onsson and A.~Rantzer, ``Systems with uncertain parameters -
  time-variations with bounded derivatives,'' \emph{33rd Conference on decision
  and control}, pp. 3074--3079, 1994.

\end{thebibliography}

\end{document}